\documentclass[12pt,a4paper]{article}
\usepackage{amssymb,amsmath,amsthm}
\usepackage{epsfig}

\newtheorem{theorem}{Theorem} 
\newtheorem{example}{Remarks}

\newcommand{\ee}{\mathbf{e}}
\newcommand {\R}{\mathbb{R}}
\newcommand {\N}{\mathbb{N}}
\newcommand {\PP}{\mathcal{P}}
\newcommand {\gU}{\mathfrak{U}}
\newcommand {\gH}{\mathfrak{H}}
\DeclareMathOperator{\Dom}{Dom}
\DeclareMathOperator{\const}{const}
\DeclareMathOperator{\id}{id}
\newcommand{\dd}{\mathrm{d}}

\date{}
\title{On the dynamics created by a time--dependent Aharonov--Bohm
  flux}
\author{J.~Asch \\
  CPT-CNRS, Luminy Case 907, F-13288 Marseille Cedex 9, France \\
  e-mail: asch@cpt.univ-mrs.fr \\[2ex]
  P.~\v{S}\v{t}ov\'{i}\v{c}ek\thanks{ Partially supported by the
    grants No.~201/05/0857 of the Grant Agency of the Czech Republic
    and
    No.~LC06002 of the Ministry of Education of the Czech Republic } \\
  Faculty of Nuclear Science, Czech Technical University,\\
  Trojanova 13, 120 00 Prague, Czech Republic \\
  e-mail: stovicek@kmalpha.fjfi.cvut.cz }

\begin{document}

\maketitle
\begin{abstract}
  We study the dynamics of classical and quantum particles moving in a
  punctured plane under the influence of a homogeneous magnetic field
  and driven by a time-dependent singular flux tube through the hole.
\end{abstract}


\section{Introduction}

The model under consideration was introduced by physicists in order to
understand the Integer Quantum Hall effect and much investigated by
mathematical physicists who introduced topological indices in order to
explain the quantization of charge transport observed in the
experiments; consult \cite{Elgart} for an access to the literature.

Let $m>0$, $e>0$, $\hbar>0$ be physical parameters,
$q\in\R^2\setminus\{0\}$, $q^{\perp}:=(-q_2,q_1)$ and $\Phi:\R\to\R$
be a smooth function. The time--dependent Hamiltonian is
\[
\frac{1}{2 m}\left(p-eA(t,q)\right)^2,\textrm{~~}
A(t,q) = \left(\frac{B}{2}
-\frac{\Phi(t)}{2\pi\vert q\vert^2}\right)q^{\perp}
\]
where in the classical case $p\in\R^{2}$ and the Hamiltonian is a
function on the phase space and where in the quantum case
$p=\left(-i\hbar\partial_{x},-i\hbar\partial_{y}\right)$ and the
Hamiltonian is the Friedrichs extension of
$(1/2m)\left(p-eA(t,q)\right)^2$ defined on
$C^{\infty}_{0}\setminus\{0\}$.

In the quantum case we discuss the meaning of the propagator and show
that an adiabatic approximation is valid. To this end we introduce the
notion of a propagator weakly associated to a time-dependent
Hamiltonian. A detailed presentation is given in
\cite{AschHradeckyStovicek}.
  
For the classical case we show: in the past the center is bound and
the particles spiral inward towards the flux line, their motion being
accompanied by energy loss; after hitting the puncture they become
``conducting'', i.e., the motion becomes a cycloid around an outward
drifting center orthogonal to the induced electric field.  The
outgoing drift is without energy loss.

The latter results have not been published yet but can be found in
preprint \cite{AschStovicek}.  Finally let us note that the dynamics
of the classical system without magnetic field was discussed in
\cite{AschBenguriaStovicek}.

\section{The quantum case}

\subsection{Existence and adiabatic approximation}

We discuss the case $\partial_{t}\Phi=\const$. After rescaling the
physical parameters and restricting ourselves to a sector of fixed
angular momentum we consider in $L^{2}((0,\infty),r\dd{}r)$ the
operator
\[
H(s) = -\frac{1}{r}\partial_{r}r\partial_{r}
+\frac{1}{r^{2}}\left(s+\frac{r^{2}}{2}\right)^{\!2}.
\]
which is essentially selfadjoint on
$C_{0}^{\infty}\left(0,\infty)\right)$ iff $\vert s\vert\ge1$, and
defined by the regular boundary condition at $r\to0$ for $\vert
s\vert<1$. We study the ``adiabatic'' limit ($\varepsilon\to 0$) of
the evolution equation
\[
i\varepsilon\partial_{s}U(s,s_{0})\psi=H(s)U(s,s_{0})\psi
\]
for the propagator $U$. Now, $\Dom(H(s))$ is time--dependent and so
the existence of a unique solution of the evolution equation is not
assured (c.f. \cite{ReedSimon2}); on the other hand $\partial_{s}H(s)$
is not relatively bounded and the gaps between the eigenvalues,
$E_{n+1}(s)-E_{n}(s)$, are approximately constant in $n$ and thus the
known theorems (c.f. \cite{AvronElgart}) do not assure the validity of
the adiabatic approximation.

Our solution to these problems is the following: we use the explicit
knowledge of the spectral measure of $H(s)$ to show the existence of
an ``adiabatic'' propagator $U_{ad}$. $U_{ad}$ in turn is used to
define a unique propagator $U_{w}$ weakly related to $H(s)$. Then we
show that $U_{ad}$ is an approximation of $U_{w}$ (see Section~2.2 for
the weak relationship).

The spectrum of $H(s)$ is discrete. Denote respectively $E_{n}(s)$,
$\psi_{n}(s)$, $P_{n}(s)$ the eigenvalues, eigenfunctions (chosen
real) and eigenprojections; let
$\PP(s):=i\sum_{\N}(\partial_{s}P_{n})P_{n}(s)$. Define $H_{ad}(s)$
and its propagator $U_{ad}$ by
\[
H_{ad}(s) := H(s)+\varepsilon\PP(s)\textrm{~~and~~}
U_{ad}(s)\psi_{n}(0) := \exp\!\left({-\frac{i}{\varepsilon}
    \int_{0}^{s}E_{n}}\right)\psi_{n}(s).
\]

\begin{theorem}{Theorem}
  For $s\ge0$,
  \begin{enumerate}
  \item $\Vert\PP(s)\Vert\le M(s)$ where $M(s)$ is a positive
    increasing function on $\R_+$,
  \item $\exists\Gamma(s)$ differentiable such that $\PP =i[H,\Gamma]$ and
    $\Vert\Gamma(s)\Vert+\Vert\partial_{s}\Gamma(s)\Vert\le\const$,
  \item 
    \[
    \left\Vert\int_{0}^{s}U_{ad}^{-1}\PP U_{ad}\right\Vert
    \le\const\ \varepsilon s,
    \]
  \item For $C(s)$ defined by
    $i\partial_{s}C(s)=-(U_{ad}^{-1}\PP{}U_{ad})(s)C(s)$, $C(0)=\id$,
    it holds
    \[ 
    \Vert C(s)-\id \Vert\le\const\ \varepsilon M(s)\exp(s M(s)).
    \]
  \end{enumerate}
\end{theorem}

\begin{proof}[Comments on the proof]
  The main problem is to control the operator bound on $l^{2}(\bf\N)$
  of the matrix
  \[
  \langle\psi_{m}, \dot{\psi}_{n}\rangle
  \sim \frac{\langle\psi_{m},\dot{H}\psi_{n}\rangle}{E_{n}-E_{m}}
  \sim \frac{1}{n-m}\left(\frac{m+1}{n+1}\right)^{s/2}.
  \]
  This is done in a number of steps. As an illustration, the first
  step is to find a bound on $L^{2}((0,\infty),\dd{}x)$ for the
  selfadjoint integral operator with the kernel
  \[
  K(x,y)=-\frac{i}{y}\left(\frac{x}{y}\right)^{s}\textrm{~~for~}x<y,
  \textrm{~~}K(x,y)=\frac{i}{x}\left(\frac{y}{x}\right)^{s}
  \textrm{~~for~}x>y.
  \]
  The bound reads $\Vert K\Vert\le(s+\frac{1}{2})^{-1}$. But more
  steps of similar nature are needed to complete the proof.
\end{proof}

\begin{example}
  \begin{enumerate}
  \item $\Dom(H_{ad}(s))=\Dom(H(s))$,
    $U_{ad}(s)(\Dom(H_{ad}(0)))=\Dom(H_{ad}(s))$,

    $i\varepsilon\partial_{s}U_{ad}(s)\psi=H_{ad}(s)U_{ad}(s)\psi,
    \quad\forall\psi\in\Dom(H(0))$.
  \item $C(s)$ is well defined by the Dyson formula.
  \item $U_{w}(s):=U_{ad}(s)C(s)$ is a propagator and the candidate to
    be generated by $H(s)$. Further it holds
    \[
    \Vert U_{w}(s)-U_{ad}(s)\Vert\le\varepsilon M(s)\exp(s M(s)).
    \]
  \item It is {an open question} whether $C(s)$ preserves $\Dom(H(0))$
    and thus whether\\ $U(s)\Dom(H(0))\subset \Dom(H(s))$.
\end{enumerate}

\end{example}

\subsection{Weakly associated propagator}

While we cannot show that the propagator $U_{w}$ is the propagator of
$H(s)$ we can show that it is the unique propagator weakly associated
to $\{H(s)\}$; so if the propagator for $H(s)$ exists, it equals
$U_{w}$.

The definition of weak association relies heavily on the notion of the
quasi-energy operator which is directly related to the propagator:
$K=\gU(-i\partial_{s})\gU^{\ast}$ where
$\gU=\int^{\oplus}U_w(s,0)\,ds$. We say that a propagator $U_w$ is
weakly associated to $H(s)$ iff
\[
K=\overline{-i\partial_{s}+\gH}\textrm{~~~where~}
\gH=\int^{\oplus}_{\R}H(s)\,ds.
\]

One can actually prove that in this way introduced notion of weak
association generalizes the standard relationship between a propagator
and a Hamiltonian as well as that at most one propagator can be weakly
associated to a Hamiltonian. For details see
\cite{AschHradeckyStovicek}.

\section{The classical case}

We again discuss the linear case $\Phi(t)=\Phi_{0}t$. After a
rescaling one is lead to consider the Hamiltonian flow of
\[
H(s) = \frac{1}{2}\left(p-a(q)\right)^{2},\textrm{~~with~}
a(q) := \left(\frac{1}{2}-{\phi}\,\frac{t}{q^{2}}\right)q^{\perp},
\]
for ${\phi}:=e\Phi_{0}/(2\pi\omega)$, $\omega:=eB/m$. Because of the
cycloid--type nature of the trajectories $q(t)$ around a moving center
$c$ we use the natural splitting $q=c+v^{\perp}$ where $v:=p-a(q)$,
$c:=q-v^{\perp}$. Let us denote
\[
\ee(\varphi):=\left(\cos\varphi,\sin\varphi\right).
\]
An appropriate canonical coordinate system is then defined so that
\[
q=\vert c\vert\ee(\varphi_{1})+\vert v\vert \ee(-\varphi_{2}).
\]
The action--angle coordinates read $I_{1}=\vert{}c\vert^{2}/2$,
$I_{2}=H$, $\varphi_{1}$, $\varphi_{2}$; the transformed Hamiltonian
is an integral of motion
\[
K(\varphi,I) = I_2-{\phi}\arg(\sqrt{2I_1}\,\ee(\varphi_1)
+\sqrt{2I_2}\,\ee(-\varphi_2)).
\]
The fundamental relation
between the center $c$ and the energy is
\[
\frac{\vert c(s)\vert^{2}}{2}=H(s)+{\phi}(s-s_{0})
\]
where $s_{0}$ is a constant depending on the trajectory. The
asymptotic behavior described below is illustrated by Fig.~1 depicting
a typical trajectory.

\begin{figure}
  \begin{center}
    \epsfig{figure=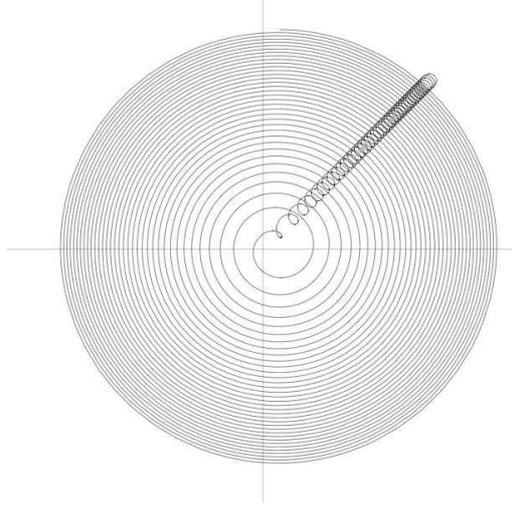,width=0.49\linewidth}
    \vskip-1.1\baselineskip
    \caption{A typical trajectory of the Hamiltonian $H(s)$}
  \end{center}
\end{figure}

\begin{theorem}{Theorem} 
  For any fixed initial condition there exists a constant $a_{0}>0$
  such that
  \begin{eqnarray*}
    && \frac{q(s)}{\sqrt{s}} \to_{s\to\infty}
    \sqrt{2{\phi}}\,\,
    \ee\!\left(\frac{a_{0}^{2}}{4{\phi}^2}-\frac{K}{{\phi}}\right),
    \textrm{~~~}
    \frac{q(s)}{\sqrt{|s|}} \sim_{s\to-\infty}
    \sqrt{2{\phi} }\,\,\ee(-s),\\
    && H(s) \to_{s\to\infty} \frac{a_{0}^{2}}{4{\phi}},
    \textrm{~~~}
    \frac{H(s)}{\vert s\vert} \to_{s\to-\infty} {\phi}.
  \end{eqnarray*}
\end{theorem}

\begin{proof}
  The problem can be reduced to a two-dimensional system with
  coordinates $J:=I_{1}+I_{2}$ and $\psi:=\varphi_{1}+\varphi_{2}$.
  After a change of variables one arrives at a system of differential
  equations which is equivalent to the integral equations
  \begin{eqnarray*}
    x_{j}(s) &=& c_{1}sJ_{j-1}(s)+c_{2}sY_{j-1}(s)\\
    && -\,\frac{\pi s}{2}\int_{s}^{\infty}
    \big(Y_{j-1}(s)J_{1}(\tau)-J_{j-1}(s)Y_{1}(\tau)\big)
    F(\tau,x_{1}(\tau),x_{2}(\tau))d\tau,
  \end{eqnarray*}
  $j=1,2$, where the numbers $c_{1}$, $c_{2}$ involve initial
  conditions and
  \[
  F(s,x_{1},x_{2}) := {\phi} -\frac{x_{1}}{s}
  -\frac{{\phi}^{2}s}{\sqrt{x_{1}^{2}+(x_{2}-{\phi})^{2}
      +{\phi}^{2}s^{2}}+x_{1}}\,.
  \]
  The integral equations allow for iterative solution and are well
  suited for asymptotic analysis.
\end{proof}

\end{document}